\documentclass[journal,twoside,web]{ieeecolor}
\usepackage{lcsys}

\usepackage{times}

\usepackage{soul}
\usepackage{url}
\usepackage[hidelinks]{hyperref}
\usepackage[utf8]{inputenc}
\usepackage[small]{caption}
\usepackage{graphicx}
\usepackage{amsmath}
\usepackage{booktabs}
\usepackage[table]{xcolor}
\definecolor{lightgray}{gray}{0.9}
\usepackage{cite}

\usepackage{amsmath,amsthm}
\usepackage{subcaption}
\usepackage{multicol}
\usepackage{amsmath,amssymb,amsfonts}

\usepackage{algorithm}
\usepackage{algpseudocode}
\usepackage{textcomp}
\usepackage{listings}
\usepackage{multirow}
\usepackage{xcolor}
\usepackage{mathtools}
\usepackage{graphicx}
\usepackage{array}
\usepackage{cancel}
\usepackage{textcomp}
\usepackage{diagbox}
\usepackage{bbm}

\let\oldReturn\Return
\renewcommand{\Return}{\State\oldReturn}

\usepackage{changes}

\usepackage{dblfloatfix}
\usepackage{times}  
\usepackage{helvet}  
\usepackage{todonotes}
\usepackage{bm}
\usepackage{tikz}
\usetikzlibrary{arrows,automata}
\usetikzlibrary{shapes,shapes.geometric}
\usepackage{float}
\usepackage{subcaption}
\usepackage{mathtools}
\usepackage{acronym}
\usepackage{courier}  
\usepackage{url}
\usepackage{colonequals}
\usepackage{graphicx}  
\usepackage{paralist}

\newtheorem{theorem}{Theorem} 
\newtheorem{lemma}[theorem]{Lemma}

\definecolor{darkgreen}{rgb}{0,0.5,0}

\allowdisplaybreaks[3] 

\usepackage{orcidlink}

\newif\ifuseboldmathops
\newif\ifuseittextabbrevs
\useboldmathopstrue   

\ifuseittextabbrevs

\newcommand{\ie}{{\it i.e.}}

\else

\newcommand{\ie}{i.e.}

\fi

\ifuseboldmathops
\newcommand{\reals}{\mathbf{R}}

\else
\newcommand{\reals}{\mathbb{R}}

\fi

\ifuseboldmathops

\else

\fi

\ifuseboldmathops

\else

\fi

\ifuseboldmathops

\else

\fi

\newcommand{\argmax}{\mathop{\mathrm{argmax}}}

\newcommand{\sink}{\mathsf{sink}}

\newcommand{\abs}[1]{\lvert#1\rvert}

\newcommand{\dist}[1]{\mathsf{Dist}(#1)}

\DeclareMathOperator*{\optmins}{\textrm{min.}}

\DeclareMathOperator*{\optsts}{\textrm{s.t.}}

\theoremstyle{definition}
\newtheorem{definition}{Definition}

\newtheorem{problem}{Problem}
\newtheorem{remark}{Remark}

\newcommand{\types}{\mathcal{T}}

\acrodef{mdp}[MDP]{Markov decision process}
\acrodefplural{mdp}[MDPs]{Markov decision processes}
\acrodef{scltl}[scLTL]{syntactically co-safe LTL}
\acrodef{dfa}[DFA]{deterministic finite-state automaton}
\acrodef{ssp}[SSP]{stochastic Shortest Path}
\acrodef{ids}[IDS]{intrusion detection system}
\acrodef{os}[OS]{operation system}
\acrodef{cvss}[CVSS]{common vulnerability scoring system}
\acrodef{milp}[MILP]{mixed-integer linear program}
\acrodef{minlp}[MINLP]{mixed-integer nonlinear program}
\acrodef{mip}[MIP]{mixed integer programming}
\acrodef{kl}[KL]{Kullback--Leibler}
\acrodef{maxent}[MAXENT]{Maximum Entropy}
\acrodef{irl}[IRL]{Inverse Reinforcement Learning}
\acrodef{ssg}[SSG]{Stochastic Stackelberg Game}
\acrodef{mtd}[MTD]{moving target defense}

\acrodef{wcarm}[WCARM-SA]{Worst-Case Absolute Regret Minimization for  Sensor Allocation} 

\title{Optimizing Sensor Allocation against Attackers with Uncertain  Intentions: A Worst-Case Regret Minimization Approach}
\author{Haoxiang Ma\orcidlink{0000-0002-3823-385X}, \IEEEmembership{Graduate Student Member, IEEE}, Shuo Han\orcidlink{0000-0003-2204-6256
}, \IEEEmembership{Member, IEEE}, Charles A. Kamhoua\orcidlink{0000-0003-2169-5975}, and Jie Fu \orcidlink{0000-0002-4470-2827}, \IEEEmembership{Member, IEEE}
\thanks{This work was sponsored in part by the Army Research Office and was accomplished under Grant
Number W911NF-22-1-0034 and W911NF2210166 and in part by NSF under grant No. 2144113. }
\thanks{Distribution A: Approved for public release. Distribution is unlimited.}
\thanks{H. Ma and J. Fu   are with the Dept. of Electrical and Computer Engineering, University of Florida, Gainesville, Fl 32611 USA.
(e-mail: hma2, fujie@ufl.edu)}
\thanks{S. Han is with the Department of Electrical and Computer Engineering, University of Illinois, Chicago, IL 60607.
(e-mail: hanshuo@uic.edu)}
\thanks{C. Kamhoua is with U.S. Army
Research Laboratory. 
(e-mail: charles.a.kamhoua.civ@mail.mil)}
}

\begin{document}

\maketitle
\thispagestyle{empty}
\pagestyle{empty}
\begin{abstract}

This paper focuses on the optimal allocation  of multi-stage attacks with the uncertainty in  attacker's intention. We model the attack planning problem using a Markov decision process and characterize the uncertainty in the attacker's intention using a finite set of reward functions---each reward represents a type of attacker. Based on this modeling,
we employ the paradigm of the worst-case absolute regret minimization   from robust game theory and develop mixed-integer linear program (MILP) formulations for solving the worst-case regret minimizing sensor allocation strategies for two classes of attack-defend interactions: one where the defender and attacker engage in a zero-sum game and another where they engage in a non-zero-sum game. We demonstrate the effectiveness of our algorithm using a stochastic gridworld example. 
\end{abstract}

\section{Introduction}
\IEEEPARstart{W}ith the increasing   severity of cyber- and physical- attacks, developing effective proactive  defense aims  to enable early detection of attacks by strategically allocating sensors/intrusion detectors. However, this task is complicated by the fact that attackers often have varying objectives and intentions. 
This paper studies the design of a robust proactive sensor allocation, given the uncertainty in the objective or the intention of the attacker. Our approach is motivated by real-world cyber security incidents,  where defenders have limited monitoring resources and must deal with attackers with different objectives, ranging from using a botnet to interrupt services with a  DDoS attack, distributing malware to steal sensitive data,  or privilege escalation attacks (see\cite{verizon2022DataBreach2022}, a report on recent cyber security incidents).

We formulate the attack planning problem  as a \ac{mdp} and enable the defender to allocate intrusion detectors, called sensors in this context, to detect the presence of an attack. The sensor allocation modifies the transition function of the attack \ac{mdp}. Specifically, when a state is allocated with a sensor, 
it becomes a sink/absorbing state, as the attack terminates once a sensor state is reached. Therefore, the goal of designing an optimal sensor allocation is to modify the transition function of the attack \ac{mdp} such that the attacker's value can be minimized given the best response attack strategy.
The sensor allocation problem in attack graphs\cite{jha2002two} is closely related to \ac{ssg} \cite{tambe2011security}. In an \ac{ssg}, the defender/leader commits to a strategy first to protect a set of targets with limited resources, while the attacker/follower selects the best response attack strategy to the defender's strategy. 
Related to \ac{ssg} for sensor allocation, Li et al.~\cite{liSynthesisProactiveSensor2023} developed a \ac{milp} formulation for solving joint   allocation of detectors and stealthy sensors that minimizes the attacker's probability of success. Sengupta et al.~\cite{senguptaMovingTargetDefense2018a} modeled   the attacker-defender interaction using a normal-form game and proposed a mixed strategy for the defender to randomize intrusion detectors. Besides the security game, other   resource allocation problems have been extensively studied. These include  distributing preventive resources to contain the spreading process \cite{nowzariOptimalResourceAllocation2017} in a network; allocating sensors for maximizing  coverage \cite{marden2016role}. The main difference is that in game-theoretic resource allocation, the decision maker's objective function is a function of the allocated resource and the best response of the attacker, which can be influenced by how the sensors are allocated. 

Traditional Stackelberg security games assume that the defender knows the attacker's payoff function, which is often not the case. To address this, researchers have studied robust defense design from the perspective of robust optimization \cite{aghassi2006robust}.  In \cite{pita2010robust}, the authors considered robust Stackelberg equilibria in   normal form games where the uncertainty comes from multiple   $\epsilon$-optimal best responses from the follower. In \cite{tan2019optimal}, the authors   used an \ac{mdp} to model the attack planning problem and employed robust optimization to design a \ac{mtd} policy that is robust to a finite uncertainty set of attack strategies.    In \cite{kroer2018robust}, the authors introduced a robust Stackelberg equilibrium that  maximizes the leader's payoff given the worst-case realization of the follower's payoff   in a deterministic sequential game in which each player selects  a distribution over action sequences. The uncertainty  is assumed to lie within a bounded interval  on the follower's payoff. 

Similar to \cite{kroer2018robust}, we also investigate the problem of robust defense when the defender has  incomplete knowledge about the attacker's payoffs, modeled as a finite set of attacker's types. Each type is associated with a unique   reward function that describes the attack objective.
We consider robust sensor allocation with the solution 
of the worst-case absolute regret minimization \cite{poursoltani2022adjustable}. 

The proposed \ac{wcarm} solution informs a regret-averse defender to choose a strategy that  leads to a small regret once he realizes what would have been
the best decision if he knew the attacker's type. As shown in operations research, the \ac{wcarm} solutions are often less conservative solutions compared to those by robust optimization \cite{savageTheoryStatisticalDecision1951, poursoltani2022adjustable}.

Our contribution can be summarized as:

\begin{itemize}
    \item We develop  \ac{wcarm} methods to solve  robust sensor allocation problems in zero-sum and non-zero-sum attack-defend interactions with uncertainty in the attack intention, described by a finite set of possible attacker's reward functions. We demonstrate that the \ac{wcarm} can be formulated as \ac{milp} problems for both cases.
    \item We leverage the zero-sum property to develop efficient solution for the \ac{wcarm} for the zero-sum case. 
    \item We validate the effectiveness of our proposed approach through   experiments with attack motion planning problems in stochastic gridworld environments.
\end{itemize}

\section{Preliminaries and Problem Formulation}
\noindent \textbf{Notations} Let $\reals$ denote the set of real numbers and $\reals^n$ 
the set of real $n$-vectors. The vector of all ones is represented as $\mathbf{1}$. The notation $z_i$ refers to the $i$-th component of a vector $z \in \reals^{n}$ or to the $i$-th element of a sequence $z_1, z_2, \ldots$, which will be clarified by the context. The set of probability distributions over a finite set $Z$ is denoted as $\dist{Z}$.

We begin by presenting an attack \ac{mdp} that captures an attacker's planning problem.

\noindent \textbf{Attack Planning Problem} 
The attack planning problem   is modeled as an attack \ac{mdp} $
M = (S, A, P, \nu, \gamma, R),
$
where $S$ is a set of states (nodes in the attack graph) including a special absorbing/sink state $s_\sink$, $A$ is a set of attack actions,  $P: S \times A \rightarrow \dist{S}$ is a probabilistic transition function such that $P(s'|s, a)$ is the probability of reaching state $s'$ given action $a$ being taken at state $s$, $\nu \in \dist{S}$ is the initial state distribution, $\gamma \in (0,1]$ is a discount factor, and $R: S \times A \rightarrow \reals$ is the attacker's reward function such that $R(s, a)$ is the reward received by the attacker for taking action $a$ in state $s$.  The attacker's objective is to maximize the total discounted rewards in the attack \ac{mdp}. For concrete examples of attack graphs generated from network vulnerabilities,   readers are directed to \cite{jha2002two} and \cite{liSynthesisProactiveSensor2023}.



We   consider Markovian policies because it suffices to search in Markovian policies for an optimal policy in the attack \ac{mdp}\cite{puterman2014markov}.  Given a Markovian policy $\pi \colon S \to \dist{A}$, the attacker's value function  $V_2 ^{\pi}: S \rightarrow \reals$ is defined as
\[
V_2^{\pi}(s) = E_{\pi}[\sum\limits_{k = 0}^{\infty}\gamma^{k}R(s_k, \pi(s_k))|s_0 =s],
\] where $E_{\pi}$ is the expectation and $s_k$ is the $k$-th state in the Markov chain induced from the \ac{mdp} $M$ under the policy $\pi$, starting from state $s$. The attacker's value given the initial distribution $\mu$ is 
$ V_2^{\pi}(\mu) =  \sum_{s\in S} \mu(s)
V_2^{\pi}(s).$

\noindent \textbf{Defender's incomplete information} The defender knows the dynamics in the attack \ac{mdp}. However, the defender does not know the exact reward function of the attacker; rather, the defender is only aware that the attacker can fall into one attack type at any given time. Different attacker types only differ in their reward function in the attack MDP and share the same states, actions, transition function, initial distribution, and discount factor.
Specifically, let $\types =\{1,\ldots, N\}$ be the attacker's type space. Let $R_i:S\times A\rightarrow \reals$ be the reward function for attacker type $i$. 

\noindent \textbf{Defender's countermeasures} To detect an ongoing attack, 
the defender is capable of allocating sensors to a subset $U\subset S$ of states in the \ac{mdp} $M$. The attack will be terminated immediately once the attacker reaches a state monitored by the sensor (assuming the sensor's false negative rate is 0) \footnote{Please see the Appendix that describes how to extend the proposed method when the sensor's false negative rate is nonzero.}. 
However, the defender's sensor allocation is constrained. Specifically, we consider a sensor allocation as a Boolean vector $\vec{x}\in \{0,1\}^{\abs{S}}$. If $\vec{x}(s)=1$, then the state $s \in U$ is allocated with a sensor. A valid allocation $\vec{x}$ needs to satisfy $\vec{x}(s)=0$ for any $s\in S \setminus U$ because only states in $U$ can be monitored.  In addition, the number of sensors cannot exceed a given integer $k$, \ie, $\mathbf{1}^{\mathsf{T}}\vec{x} \le k$.


We state the problem informally as follows. 
\begin{problem}
\label{pro: Optimization}
In the attack planning modeled as the \ac{mdp} $M$ with uncertainty in the attacker's type,  how to robustly allocate limited sensors with respect to the defense objective?
\end{problem}

\section{Main Results}

First, it is observed  that a sensor allocation changes the transition function of the attack \ac{mdp} as follows.

\begin{definition}[Attack \ac{mdp} equipped with sensors]
\label{def:mdp-sensors}
Given a sensor allocation $\vec{x}$ and the original attack \ac{mdp} $ M =(S,A,P, \nu, \gamma, R)$,  the attack \ac{mdp} under $\vec{x}$ is the \ac{mdp} 
\[
M(\vec{x}) = (S,A, P^{\vec{x}}, \nu, \gamma, R),
\]
where $S,A,\nu, \gamma, R$ are identical to those in $M$, and $P^{\vec{x}}$ is defined as
\begin{equation*}
\label{eq:def_px}
P^{\vec{x}}(s'|s,a)=
\begin{cases}
	1, & \vec{x}(s)=1,\ s'=s_{\sink},\\
	0, & \vec{x}(s)=1,\ s'\neq s_{\sink},\\
	P(s'|s,a), & \vec{x}(s)=0.
\end{cases} 
\end{equation*}
\end{definition}

To allocate sensors with  uncertainty in the attacker's type, we employ a solution of robust game \cite{poursoltani2022adjustable,stoye2011statistical}, called \emph{worst-case absolute regret minimization}, which 
optimizes the performance of a decision variable, $\vec{x}$ in our context, with respect to the ``worst-case regret'' that might be experienced when comparing $\vec{x}$ to the best decision that should have been made given  the attacker's type $i$ is known. 
Next, we discuss two approaches to solve worst-case regret minimizing sensor allocations given \begin{inparaenum}[i)]
\item  \emph{Zero-sum attack-defend game}: For each attack type $i$,  the defender's value in $M(\vec{x})$ is the negation of the attacker $i$'s value.
\item \emph{Non-zero-sum attack-defend game}: 
Regardless of the attacker's type, the defender's value in $M(\vec{x})$ is defined by evaluating the attacker's strategy with respect to a defender's cost function $C:S\times A\rightarrow \reals$. The defender's goal is to minimize the total discounted costs respecting $C$ incurred by the attacker's strategy.
\end{inparaenum}

\subsection{Worst-case regret minimization in zero-sum game}
\label{sec: worst-case zero sum}

In the zero-sum case, it is first noted that the optimal sensor allocation $\vec{x}_i$ for attacker's type $i$ can be obtained by solving the following optimization problem:
\[
\vec{x}_i   =\arg \min_{\vec{x} \in \mathcal{X}}\max_{\pi} V_{2, i}^{\pi}(\nu; \vec{x}).
\]

where $V_{2, i}^{\pi}(\nu; \vec{x})$ is attacker $i$'s value given attack strategy $\pi$ in the \ac{mdp} $M(\vec{x})$ and the attacker $i$'s reward $R_i$.
The optimal sensor allocation problem can be formulated as an \ac{milp} (see \cite{liSynthesisProactiveSensor2023}). 
As a result, the \ac{wcarm} problem  is formulated as:
\[
(\mbox{\ac{wcarm}}) \mbox{minimize}_{\vec{x}\in \mathcal{X}}\max_{i\in \types} \left(  V_{2, i}(\vec{x})  - V_{2, i}(\vec{x}_i) \right)
\]
 where $V_{2, i}(\vec{z})  = \max_\pi  V_{2, i}^{\pi}(\nu, \vec{z})$ for $z = \vec{x}, \vec{x}_i$. That is, the attacker always chooses the optimal strategy in \ac{mdp} $M(\vec{x})
$.  The difference $V_{2, i}(\vec{x})  - V_{2, i}(\vec{x}_i)$ measures the regret of the defender  for choosing $\vec{x}$ instead of $\vec{x}_i$ when the attacker is type $i$. The regret  is always non-negative for any sensor allocation decision $\vec{x} \in \mathcal{X}$ because $\vec{x}_i =\arg\min_{\vec{x}} V_{2, i}(\vec{x})$.

 Because $\vec{x}_i$ is pre-computed for each attacker type, the quantity $V_{2, i}(\Vec{x_i})$ is a constant, denoted by $v_i$ for clarity. The optimization problem is then written as:

\[ \quad \mbox{minimize}_{\vec{x}\in \mathcal{X}}\max_{i\in \types} \left( V_{2, i}(\vec{x})  - v_i  \right),
\]
which is a robust optimization problem. 
The following lemma shows how the robust optimization problem can be reformulated as an \ac{milp}.
\begin{lemma}

\label{lemma: zero-sum case lemma}
	The worst-case absolute regret minimization problem for robust sensor allocation in a zero-sum game is equivalent to the following optimization problem: 
%
%
  \vspace{-1ex}
  \begin{alignat}{2}
  \label{eq:zerosum-regret}
  	& \optmins_{ y, \vec{x} \in\mathcal{X}} && \quad y\\
  	& \optsts &&  y \ge   V_{2, i}(\vec{x}) - v_i, \forall i \in \types, \label{eq:regret-constraint} \\
  	& && V_{2, i}(\vec{x}) = \sum_{s\in S} \nu(s)V_{2, i}(s;\vec{x}), \forall i \in \types,  \label{eq:regret-constraint-2}\\ 
  	& && V_{2, i}(s;\vec{x}) \ge R_i(s,a ) + \gamma \sum_{s'} P^{\vec{x}}(s'|s, a) V_{2, i}(s'; \vec{x}), \nonumber \\
  	& && \forall s\in S, \forall a \in A, \forall i \in \types, \label{eq:regret-constraint-3}\\
   & && \mathbf{1}^{\mathsf{T}}\vec{x} \le k.\label{eq:robust_constraint_lp}
  \end{alignat}
\end{lemma}

\begin{proof} 
	For attacker type $i$ and a sensor design $\vec{x}$, 
the optimal attacker's value vector, denoted 
$V^\ast_{2, i}(\vec{x})\in \reals^{\abs{S}}$ satisfies the Bellman optimality condition: For all $s\in S$, 
\[
V^\ast_{2,i}(s, \vec{x}) =\max_{a\in A} \left(R_i(s,a) + \gamma \sum_{s'\in S } P^{\vec{x}}(s'|s, a) V^\ast_{2,i}(s'; \vec{x})\right).
\]
\vspace{-1ex}

Based on the linear program formulation of dynamic programming \cite{defariasLinearProgrammingApproach2003}, any  vector $V_{2,i}(\vec{x})$ satisfying the set of constraints  in \eqref{eq:regret-constraint-3} is an upper bound on the $V^\ast_{2,i} (  \vec{x})$,  for all $s\in S$ element-wise. Therefore, $\sum_{s\in S} \nu(s) V_{2,i}(s;\vec{x}) \ge \sum_{s\in S} \nu(s) V^\ast_{2,i}(s; \vec{x})$. Constraints  \eqref{eq:regret-constraint},  \eqref{eq:regret-constraint-2}, and \eqref{eq:regret-constraint-3} together enforce $y \ge \max_{i \in \types} (\sum_{s\in S} \nu(s) V_{2, i}(s; \vec{x})  - v_i ) \ge  \max_{i \in \types} ( V^\ast_{2, i}(\nu;\vec{x}) - v_i)$.

For an arbitrary $\vec{x}$,  let $r(\vec{x}) =\argmax_{i\in \types} ( V^\ast_{2,i}(\nu;\vec{x}) - v_i)$,
that is, the $r(\vec{x})$ is the attacker type for which the defender's regret of using $\vec{x}$ is the largest among the regret for all attacker's types.
Then we have $y \ge \max_{i\in \types} ( V^\ast_{2,i}(\nu;\vec{x}) - v_i)  =  V^\ast_{r(\vec{x})}(\nu;\vec{x}) - v_{r(\vec{x})}$.  Because $v_{r(\vec{x})}$ is a  constant once ${r(\vec{x})}$ is determined,   minimizing $y$ is equivalent to minimizing the upper bound  of $V_{r(\vec{x})}(\nu;\vec{x})$ and thus $y = V^\ast_{r(\vec{x})}(\nu;\vec{x})-v_{r(\vec{x})}$. The optimization problem is then $\min_{\vec{x}} (V^\ast_{r(\vec{x})}(\nu;\vec{x})-v_{r(\vec{x})})$, which is  equivalent to the \ac{wcarm} formulation.
\end{proof}

It is noted that the optimization problem in \eqref{eq:zerosum-regret} is nonlinear because the transition function $P^{\vec{x}}$ depends on the decision variable $\vec{x}$ and the constraints in \eqref{eq:regret-constraint-3} 
include   product terms between the transition probabilities with the other  variable $V_{2,i}(s,\vec{x})$. We show how to transform the nonlinear program into an \ac{milp} next.

By the definition of $P^{\vec{x}}$ in~\eqref{eq:def_px},  the term $\sum_{s'}P^{\vec{x}}(s'|s,a)V_{2,i}(s';\vec{x})$ in~\eqref{eq:robust_constraint_lp} satisfies
  \vspace{-1ex}
\begin{align*}
 & \sum_{s'}P^{\vec{x}}(s'|s,a)V_{2,i}(s';\vec{x})\\
& \quad=\begin{cases}
V_{2,i}(s_{\sink};\vec{x}), & \vec{x}(s)=1,\\
\sum_{s'}P(s'|s,a)V_{2,i}(s';\vec{x}), & \vec{x}(s)=0
\end{cases}\\
& \quad=\sum_{s'}P(s'|s,a)V_{2,i}(s';\vec{x})(1-\vec{x}(s))+V_2(s_{\sink};\vec{x})\vec{x}(s)\\
& \quad=\sum_{s'}P(s'|s,a)V_{2,i}(s';\vec{x})(1-\vec{x}(s)),
\end{align*}
  \vspace{-1ex}
where the last equality is implied by $V_2(s_{\sink};\vec{x}) = 0$. Define
\begin{equation}
\label{eq:def_W}
\resizebox{0.91\hsize}{!}{%
$W_{2,i}(s,s')=V_{2,i}(s';\vec{x})(1-\vec{x}(s))  =\begin{cases}
V_{2,i}(s';\vec{x}), & \vec{x}(s)=0,\\
0, & \vec{x}(s)=1.
\end{cases}$
}
\end{equation}
Using the big-M method~\cite{griva2009linear}, Eq.~\eqref{eq:def_W} can be expressed equivalently as affine inequalities (in $\vec{x}$, $V_{2,i}$, and $W_{2,i}$)
  \vspace{-1ex}
\begin{align}
& W_{2,i}(s, s') \leq M(1-\vec{x}(s)), \label{eq:constraint2}  \\
& W_{2,i}(s, s') \geq m(1 - \vec{x}(s)), \label{eq:constraint3} \\
& W_{2,i}(s, s') - V_{2,i}(s' ;\vec{x}) \leq M\vec{x}(s), \label{eq:constraint4} \\
& W_{2,i}(s, s') - V_{2,i}(s' ;\vec{x}) \geq m\vec{x}(s). \label{eq:constraint5}
\end{align}

with proper choices of constants $M>0$ and $m<0$. For example, let $M$ be the upper bound on the total rewards and $m$ be the negation of the upper bound on the total rewards.






 

\subsection{Worst-case regret minimization in non-zero-sum game}
\label{sec: optimal against one general sum}
Next, we consider the scenario where the defender aims to minimize her cost function  $C: S\times A \rightarrow \reals$, which maps a state and an attack action to a cost penalty that measures the loss incurred by the attacker's action. Because the penalty is not necessarily the negation of the attack reward, the attack-defend game is non-zero-sum. 
For a single type of attacker,
we formulate the problem as a Stackelberg game as follows.

\begin{definition}
\label{def:ssg}
For a single type of attacker whose attack \ac{mdp} is  $M=(S,A, P, \mu, \gamma, R )$, and the defender's capability of allocating sensors,  an  \ac{ssg} is formulated as a tuple $G = (S, A_1, A_2, \mathcal{P}, \nu, \gamma, R_1, R_2),$ where 
\begin{itemize}
    \item $S$, $\nu$, $\gamma$ are the same components in  the attack \ac{mdp} $M$ for states, initial distribution, and discount factor.
    \item $A_1= \{0,1\}$
 is the defender/leader's action set. Action $0$ for not allocating a sensor, 1 for allocating a sensor.
 \item $A_2 = A$ is the attacker/follower's actions.
 \item $\mathcal{P}(s'|s, a_1, a_2)$ is the probability of reaching state $s'$ given action $a_1,a_2$ being taken by the defender and the attacker at state $s$. For a state $s\in S$, a defender's action $a_1\in A_1$ and an attacker's action $a_2\in A_2$, let
\begin{equation*}
\label{eq:def_px}
\mathcal{P}(s'|s,a_1, a_2)=
\begin{cases}
	1, & a_1 = 1,\ s'=s_{\sink},\\
	0, & a_1 = 1,\ s'\neq s_{\sink},\\
	P(s'|s, a_2), & a_1 = 0.
\end{cases} 
\end{equation*}
 \item $R_1: S \times A_1 \times A_2 \rightarrow \reals$ (resp. $R_2: S \times A_1 \times A_2 \rightarrow \reals$) is the leader(resp. follower)'s reward function,  defined with the defender's cost $C$ (resp. the attacker's reward $R$): \begin{equation*}
\label{eq:def_px}
R_1(s, a_1, a_2)=
\begin{cases}
	0, & a_1=1,\\
	- C(s, a_2), & a_1=0,\\
\end{cases} 
\end{equation*}

\begin{equation*}
\label{eq:def_px}
R_2(s, a_1, a_2)=
\begin{cases}
	0, & a_1=1,\\
	R(s, a_2), & a_1=0.\\
\end{cases} 
\end{equation*}
  \end{itemize} 
  \end{definition}
In this \ac{ssg}, the defender/leader decides the sensor allocation, which determines the transition function $\mathcal{P}$. The attacker/follower decides on the best response to maximize his reward. The reward function is understood as follows: When the defender allocates a sensor to state $s$ (i.e., $a_1=1$), the attack terminates in that state, and no further costs/rewards will be incurred for either player. Otherwise (i.e., $a_1=0$), the attacker continues to reach the next state with action $a_2$, and both the defender and the attacker receive a reward of $-C(s,a_2)$ and $R(s,a_2)$, respectively.
In this formulated \ac{ssg}, both the defender and the attacker aim to maximize their respective total discounted rewards.



Since    sensors cannot be moved once allocated, we restrict the defender's strategy to be deterministic and memoryless. Given a fixed sensor allocation, the best response attack strategy can also be deterministic.
For a strategy profile $(\pi_1,\pi_2)$--a tuple of the defender's strategy $\pi_1: S\rightarrow A_1$ and $\pi_2: S\rightarrow A_2$, the defender's value function $V_1$ is defined as 
\begin{align*}
\resizebox{1.0 \linewidth}{!} 
{$V_1(s; \pi_1, \pi_2) = E_{(\pi_1,\pi_2)}[\sum\limits_{k = 0}^{\infty}\gamma^{k}R_1(s_k, \pi_1(s_k), \pi_2(s_k))|s_0 =s]$, }
\end{align*}
%
%
%
%
where the expectation is taken in the Markov chain induced from $G$ given the strategy profile $(\pi_1,\pi_2)$.

For the non-zero-sum case,      the \ac{wcarm} problem takes the following form:
\[
(\mbox{\ac{wcarm}}) \mbox{minimize}_{\pi_1}\max_{i\in \types} \left(  V_{1, i}(\nu, \pi_{1}^{i}) - V_{1, i}(\nu, \pi_1) \right),
\] 
where $  V_{1, i}( \nu, \pi_1^i) $
is the   defender's value given both the defender and the attacker committing to the Stackelberg equilibrium in the \ac{ssg} $G$ (Def.~\ref{def:ssg}) where the attacker's reward $R_2$ is defined based on the reward function of attack type $i$. The Stackelberg equilibrium can be solved with methods in \cite{vorobeychik2012computing},
with a  modification that constrains the defender's strategy to be deterministic. The regret $ V_{1, i}(\nu, \pi_{1}^{i}) - V_{1, i}(\nu, \pi_1)$ measures the defender's regret in using strategy $\pi_1$ against attacker $i$ to the defender's best strategy $\pi_1^i$ that should have been employed when playing against attacker $i$.

To find $\pi_1^\ast$ that minimizes the worst-case regret, we introduce  a decision variable $y$ and rewrite the optimization problem as follows:
\vspace{-1ex}
\begin{alignat}{3}
	& \optmins_{ \pi_1 } && \quad  y\\
	& \optsts &&  y \ge   V_{1, i}(\nu, \pi_1^i) - V_{1, i}(\nu, \pi_1), \forall i \in \types ,
	 \label{eq:best-response-i}
\end{alignat}
where $V_{1,i}(\nu, \pi_1^i)$ is a constant and is denoted by $\bar v_{1, i}$. 

We extend the \ac{milp} formulation in \cite{vorobeychik2012computing} and get the following \ac{minlp} formulation to solve the \ac{wcarm} problem:
\vspace{-1ex}
\begin{subequations}
\label{eq: non-zero-sum-constraints}
\begin{align}
    &\min\limits_{\pi_1,  \{V_{1, i}, \pi_2^i, V_{2, i}\mid i\in \mathcal{T}\}}  y \\
    &\text{subject to: } \nonumber \\
    \label{constraint: defender regret}
    & y \ge  \bar v_{1, i}- \sum_{s \in S}\nu(s)V_{1, i}(s),\; \; \forall i \in \types, \\
    \label{constraint: defender policy 1}
    &\pi_1(a_1 \mid s) \in \{0, 1\},  \; \forall s \in S, a_1 \in \{0, 1\}, \\ 
    \label{constraint: defender policy 2}
    &\sum\limits_{a_1 \in \{0, 1\}} \pi_1(a_1 \mid s) = 1 \; \; \forall s \in S,\\
    \label{constraint: attacker policy 1}
    &\pi_2^i(a_2 \mid s) \in \{0,1\}, \; \; \forall i \in \types, s \in S, a_2 \in A, \\ 
    \label{constraint: attacker policy 2}
    &\sum\limits_{a_2 \in A} \pi_2^i(a_2 \mid s) = 1, \; \; \forall i \in \types, s \in S, \\
    \label{constraint: sensor number}
    & \sum_{s \in S} \pi_1(1 | s) \le k, \\
    & \text{The following constraints hold $\forall$ $ i \in \types, s \in S, a_2 \in A_2$:} \nonumber \\
    \label{constraint: worst-case regret defender}
    &(\pi^i_2(a_2 \mid s) - 1)Z \le V_{1, i}(s) - \tilde R_1^i(s, \pi_1, a_2) \nonumber \\
    & \le (1-\pi^i_2(a_2 \mid s))Z, \\
    \label{constraint: worst-case regret attacker}
    &0 \le V_{2, i}(s) - \tilde R_2^i(s, \pi_1, a_2) \le (1-\pi^i_2(a_2 \mid s))Z,
\end{align}
\end{subequations}

where 
  $\tilde R_1^i(s,\pi_1, a_2) = \sum_{a_1 \in \{0, 1\}}\pi_1(a_1|s)(R_1(s, a_1, a_2) + \gamma \sum_{s'}\mathcal{P}(s'|s, a_1, a_2)V_{1,i}(s'))$ represents the defender's expected value from state $s$ given the sensor allocation $\pi_1$ and attacker's action $a_2$ from state $s$, and the function $\tilde R_2^i(s, \pi_1,a_2)$ is defined for the attacker analogously by substituting the defender's reward and value of the next state with the attacker's. $Z$ is a large constant number, which can be the upper bound on the absolute value of total rewards.

  The constraints in \eqref{eq: non-zero-sum-constraints} are explained as follows: 
Constraint~\eqref{constraint: defender regret} enforces $y$ to be the worst-case regret. Constraint~\eqref{constraint: defender policy 1} and \eqref{constraint: defender policy 2} enforce the defender takes a deterministic strategy, constraint~\eqref{constraint: attacker policy 1}, \eqref{constraint: attacker policy 2} enforce attacker takes a deterministic strategy. Constraint~\eqref{constraint: sensor number} enforces the defender can not allocate more than $k$ sensors. Constraint~\eqref{constraint: worst-case regret defender} enforces that when the attacker $i$ takes action $a_2$, the defender's value  at that state $V_{1,i}(s)$ should equal the expected value $\tilde R_1^i(s, \pi_1,a_2)$ of defender against attacker $i$ who takes action $a_2$ at state $s$. When the attacker does not take action $a_2$, the constraint is non-binding. This set of constraints obtain $V_{1,i}(s)$ by evaluating strategy $\pi_1$ at the state $s$ against the attacker's best response for $\pi_1$.
Constraint~\eqref{constraint: worst-case regret attacker} enforces that when the attacker $i$ takes action $a_2$, his value $V_{2,i}(s)$ should be   the same as his expected value $\tilde R_2^i(s, \pi_1,a_2)$ given that action $a_2$. When $a_2$ is not taken at $s$, the attacker's value $V_{2,i}(s)$  should be greater than the expected value $\tilde R_2^i(s, \pi_1,a_2)$  due to the fact that $V_{2,i}(s)  = \max_{a_2}\tilde R_2^i(s, \pi_1,a_2)$. By enforcing this constraint, we can ensure $V_{2,i}$ is the attacker's value  given the defender's strategy $\pi_1$ and the attacker's best response to $\pi_1$.

 

 \begin{lemma}
	The worst-case absolute regret minimization solution for robust sensor allocation in the non-zero-sum attack-defender game is equivalent to the solution of ~\eqref{eq: non-zero-sum-constraints}.
\end{lemma}
The proof is similar to that of Lemma~\ref{lemma: zero-sum case lemma} and can be found in the Appendix. 
 The key insight is that constraints~\eqref{constraint: defender regret}, \eqref{constraint: worst-case regret attacker} and\eqref{constraint: worst-case regret defender} together enforce $y \ge \max_{i \in \types} (\bar v_{1, i} - \sum_{s\in S} \nu(s) V_{1, i}(s))$ where $V_{1,i}(\cdot)$ is the defender's value given the best response of attacker type $i$ (enforced by \eqref{constraint: worst-case regret defender}).

 The above formulation is nonlinear due to the interaction between the integer variable $\pi_1$ and the continuous variable $V_{1, i}$ in $\tilde R_1^i$ ($\tilde R_2^i$ analogously). 
 But since the integer variable  is binary, we can use McCormick Relaxation\cite{mitsos2009mccormick} to reformulate \ac{minlp} into \ac{milp}.
 To do so, let's introduce new variables for the defender: for $i \in \types, s\in S, a_1\in \{0,1\}, a_2\in A_2$, define $w_{s, i}^{a_1, a_2} = \pi_1(a_1 | s)\sum_{s'}\mathcal{P}(s'|s, a_1, a_2)V_{1, i}(s')$. Analogously 
 let $z_{s, i}^{a_1, a_2}$ be defined for the attacker. 
 
 We then replace $\tilde R^i_1(s, \pi_1, a_2)$ in \eqref{constraint: worst-case regret defender} with $\sum_{a_1 \in \{0, 1\}}(\pi_1(a_1 | s)R_1(s, a_1, a_2)+\gamma w^{a_1, a_2}_{s, i})$ and add the following constraints: $\forall i \in \types, s\in S, a_1 \in \{0, 1\}, a_2 \in A$:
 \vspace{-3ex}
 \begin{subequations}
 	\label{straincomponent}
 	\begin{align}
 		w_{s, i}^{a_1,a_2} \geq \sum\limits_{s' \in S}\mathcal{P}(s'|s, a_1, a_2)V_{1, i}(s') - Z(1-\pi_1(a_1|s)), \nonumber \\
 		w_{s, i}^{a_1,a_2} \leq \sum\limits_{s' \in S}\mathcal{P}(s'|s, a_1, a_2)V_{1, i}(s') + Z(1-\pi_1(a_1|s)), \nonumber \\
 		-Z\cdot \pi_1(a_1|s) \leq w_{s, i}^{a_1,a_2} \leq Z \cdot \pi_1(a_1|s). \nonumber
 	\end{align}
 \end{subequations} 
We replace $\tilde R^i_2(s, \pi_1, a_2)$ in \eqref{constraint: worst-case regret attacker} and add constraints for $z^{a_1, a_2}_{s, i}$ analogously. 
For both zero-sum and non-zero-sum cases, the formulated \ac{milp} can be solved using the Gurobi Solver.

\begin{remark}
If different attackers have different transition functions $P_i$, then  their corresponding transition functions are used in   Constraints \eqref{eq:regret-constraint-3} for the zero-sum case and in  defining $\tilde  R^i_1(s, \pi_1, a_2)$ and $\tilde R^i_2(s, \pi_1, a_2)$ for the non-zero-sum case.
\end{remark}

\noindent \textbf{Complexity analysis}: Solving an \ac{milp} is NP-complete   and its runtime complexity depends on the number of constraints and integer variables. In the   non-zero-sum case, the number of integer variables required for the \ac{wcarm}  is $O(\abs{S}\times \abs{A_2} \times \abs{\mathcal{T}})$. The number of constraints is $O(\abs{S}\times \abs{A_1}\times \abs{A_2}\times \abs{\mathcal{T}})$. For the zero-sum case, the number of integer variables and constraints is $O(\abs{S})$ and $O(\abs{S} \times \abs{A_2}\times \abs{\mathcal{T}})$ respectively. While the \ac{wcarm} solution for the non-zero-sum case can be applied to the zero-sum case, the zero-sum case formulation   in Sec.~\ref{sec: worst-case zero sum} is more efficient.  
\section{Experiments}

\begin{figure}[ht]
 \vspace{-2ex}
	\centering
	\includegraphics[width = 0.5\linewidth]{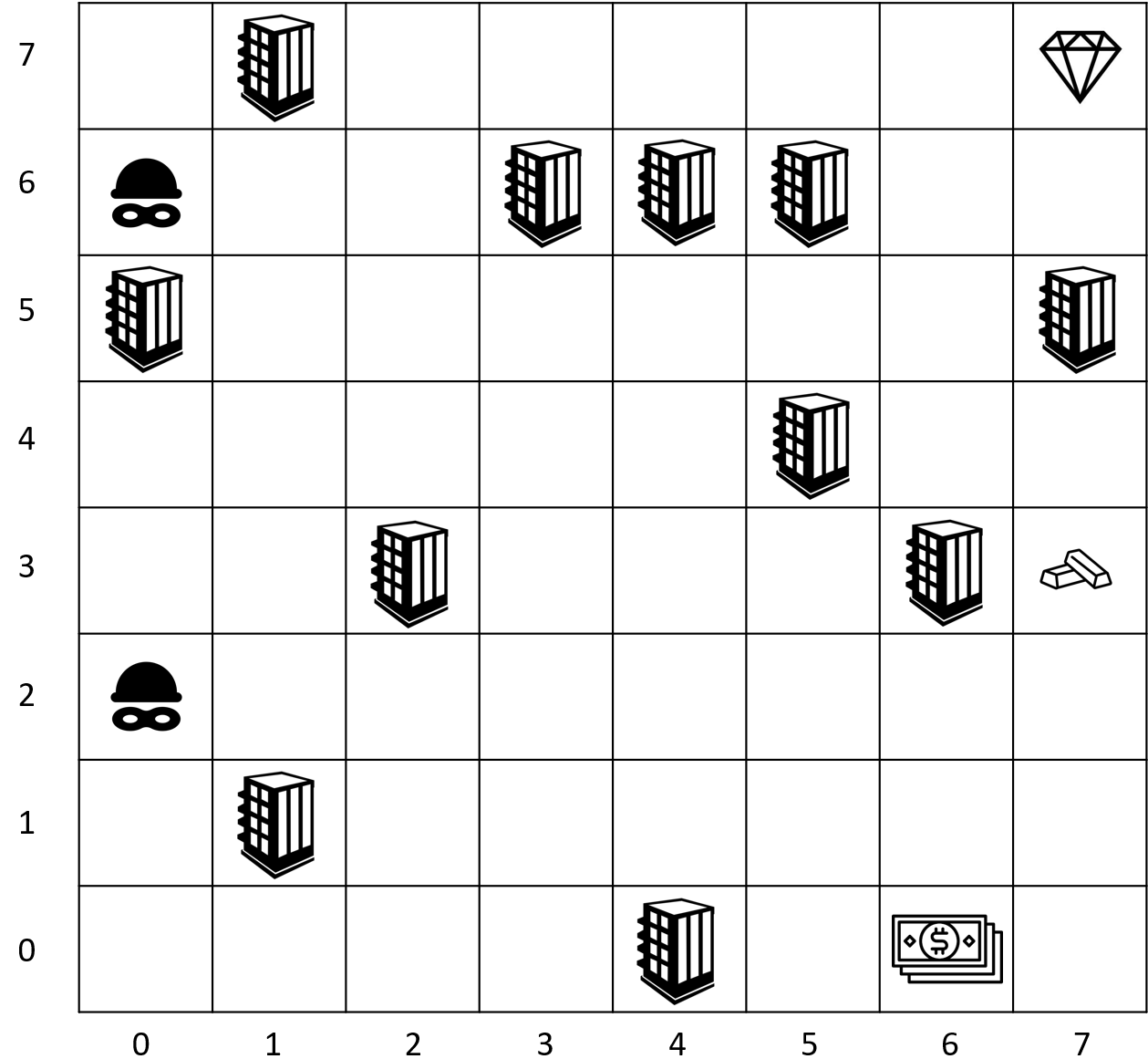}
	\caption{The 8 $\times$ 8 Gridworld Example.}
	\label{fig:8X8 gridworld}
 \vspace{-1ex}
\end{figure}

\begin{table}[ht]
\centering
 \vspace{-2ex}
\resizebox{\columnwidth}{!}{%
\begin{tabular}{|l|l|l|l|}
\hline
Reward/Cost  & Cash (0, 6) & Gold (3, 7) & Diamond (7, 7) \\ \hline
Attacker  1 & 15          & 12          & 12             \\ \hline
Attacker  2 & 12          & 15          & 15             \\ \hline
Defender        & 15          & 10          & 10            \\ \hline
\end{tabular}
}
\caption{Reward functions of two attackers and the cost function of the defender (for the non-zero-sum case).}
\label{tab: reward table}\vspace{-4ex}
\end{table}
We used an $8 \times 8$ gridworld environment,  depicted in Figure~\ref{fig:8X8 gridworld}, to demonstrate our solutions. A state is denoted by $(\mbox{row}, \mbox{col})$. There are two types of attackers with the same initial state distribution. Both attackers have a $70\%$ probability of starting from state $(2, 0)$ and a $30\%$ probability of starting from state $(6, 0)$. Each attacker can move in one of four compass directions. When given the action ``N'', the attacker enters the intended cell with a $1 - 2\alpha$ probability, and enters the neighboring cells, which are the west and east cells, with   probability $\alpha$. In our experiment, we set $\alpha = 0.1$.   If the attacker moves into the  buildings or the boundary, he remains in the previous cell. The environment contains three final states, each with a different value for the attackers. The attackers only receive a reward when they reach these final states. Table~\ref{tab: reward table} lists the rewards for the attackers.  
 
First, we consider the worst-case regret minimization in the zero-sum case. We change the number of sensors the defender can allocate to evaluate how sensor numbers affect the defender's value. As shown in Figure~\ref{fig: attacker's value},  the attacker's expected value decreases when the sensor number increases. 

In the case where only two sensors can be allocated, we compare the robust sensor allocation $\vec{x}^\ast$ with the optimal allocation strategies $\vec{x}_1$ and $\vec{x}_2$ against attacker types $1$ and $2$, respectively. When the defender knows the attacker's type, $\vec{x}_1$ yields a value of  {$5.96$} for attacker $1$ and $\vec{x}_2$ yields a value of {$6.56$}  for attacker $2$. However, when the defender does not know the attacker's type and implements $\vec{x}^\ast$, the values become {$6.12$ and $7.04$}  for attackers $1$ and $2$, respectively. In this case, the defender's worst-case absolute regret is {$\max(6.12 - 5.96, 7.04 - 6.56) = 0.48$}.

Using $\vec{x}_1$ against attacker $2$ results in a regret of $0.89$, while using $\vec{x}_2$ against attacker $1$ results in a regret of $1.49$. In both cases, the worst-case regret is higher than that of $\vec{x}^\ast$, indicating the advantage of the \ac{wcarm} method. 
The \ac{wcarm} sensor allocations for different numbers of sensors are listed in Table~\ref{tab: sensor allocation zero-sum} as well as the defender's worst-case regret under the robust policy. When 4 sensors are allocated, $\vec{x}_1=\vec{x}_2= \vec{x}^\ast = (1, 5),(2, 5),(5, 5),(7, 5)$  and the defender's worst-case regret $y=0$. An attacker is prevented from reaching any final state with probability $1$.

\vspace{-1ex}
\begin{table}[ht]
\centering
\resizebox{\columnwidth}{!}{%
\begin{tabular}{|l|l|l|l|l|}
\hline
$k$ & Policy $\vec{x}_1$ &   Policy $\vec{x}_2$ & Robust policy $\vec{x}^\ast$ &  $y$\\ \hline
1                 & $(1, 5)$                                   & $(2, 2)$                                   & $(2, 3)$     & $0.22$                                                                                                \\ \hline
2                 & $(1, 5),(2, 4)$                                 & $(2, 2),(6, 2)$                                 & $(1, 5)(6, 2)$       & $0.48$                                                                                            \\ \hline
3                 & $(1, 5),(2, 4),(6, 2)$                               & $(2, 3)(5, 2)(6, 2)$                               & $(1, 5),(2, 2)(6, 2)$    & $0.77$                                                                                             \\ \hline
\end{tabular}
}
 \vspace{-1ex}
\caption{Sensor allocation in the zero-sum case.  $\vec{x}_i$ is the optimal defender policy against attacker type $i\in \{1,2\}$.}
\label{tab: sensor allocation zero-sum}
\end{table}

\begin{figure}[ht]
 	\centering
  \vspace{-2ex}\includegraphics[width = 0.65\linewidth]{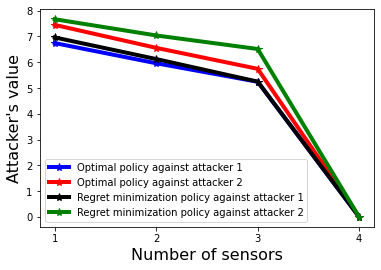}
 \vspace{-1ex}
	\caption{The attacker's expected value for zero-sum case.}
	\label{fig: attacker's value}   \vspace{-2ex}
\end{figure}


Moving on to the non-zero-sum case, the defender now receives a penalty when the attacker reaches the goal state, and the cost function for the defender is listed in Table~\ref{tab: reward table}. The sensor allocation strategies for the non-zero-sum cases, as well as the defender's worst-case regret under the robust policy, can be found in Table~\ref{tab: sensor allocation general-sum}.  
Noted that the worst-case regret is not indicative of the effectiveness of sensor allocations, that is, a small regret does not necessarily mean a large value for the defender. Thus, Figure~\ref{fig: defender's value} demonstrates that the defender's expected value increases as the number of sensors   increases.  
From Fig.~\ref{fig: defender's value}, it is observed that when two sensors are deployed, the worst-case regret is the largest  for both attackers.  

Similar to the zero-sum case, when 4 sensors can be allocated, $\vec{x}_1=\vec{x}_2= \vec{x}^\ast = (1, 5),(2, 5),(5, 5),(7, 5)$  and the defender's worst-case regret $y=0$.  

\begin{table}[ht]
\centering
\resizebox{\columnwidth}{!}{%
\begin{tabular}{|l|l|l|l|l|}
\hline
$k$ & Policy $\pi_1$ & Policy $\pi_2$ & Robust Policy $\pi^\ast$  & $y$\\ \hline
1                 & $(1, 5)$                                   & $(2, 2)$                                   & $(2, 3)$         & $0.26$                                                                                            \\ \hline
2                 & $(1, 5), (2, 4)$                                 & $(2, 3), (6, 2)$                                 & $(1, 5), (7, 5)$      & $1.19$                                                                                            \\ \hline
3                 & $(1, 4),(2, 5),(6, 2)$                               & $(2, 3),(5, 2),(6, 2)$                               & $(1, 5),(2, 2),(6, 2)$  & $0.20$                                                                                                 \\ \hline
\end{tabular}
}
\vspace{-1ex}
\caption{Sensor allocation in the non-zero-sum case.  $\vec{x}_i$ is the optimal defender policy against attacker type $i\in \{1,2\}$.}
\label{tab: sensor allocation general-sum}\vspace{-3ex}
\end{table}

\vspace{-2ex}
\begin{figure}[ht]
	\centering
	\includegraphics[width = 0.65\linewidth]{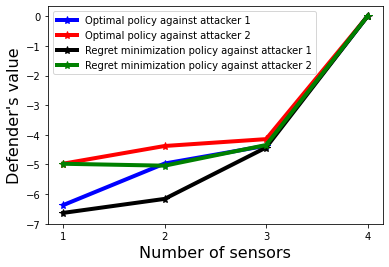}
  \vspace{-1ex}
	\caption{The defender's expected value.}
	\label{fig: defender's value} \vspace{-1ex}
\end{figure}

The experiments are conducted on a Windows 10 machine with Intel(R) i7-11700k CPU and $32$ GB RAM. The computation time for robust sensor allocations in the zero-sum cases is less than 1 second. For non-zero-sum cases, it takes from 77 sec to 3 hours  to solve given increasing sensor numbers.


\section{Conclusion}
We develop   robust sensor allocation methods in probabilistic attack planning problems using worst-case absolute regret minimization from robust game theory. We   demonstrated that both robust zero-sum and non-zero-sum sensor allocation problems can be formulated as \ac{milp}s. Our approach is suitable for a wide range of safety-critical scenarios that involve constructing probabilistic attack graphs from known network vulnerabilities.  Future work could focus on developing more efficient and approximate solutions for robust sensor allocations in non-zero-sum  games. Additionally, the solution concept for robust games can be extended to design moving target defenses that randomize network topologies and the integrated design of sensor allocation and moving target defense.





\bibliographystyle{plain}
\bibliography{refs.bib}
\section*{Appendix}
\subsection{Proof of Lemma 2}
For attacker type $i$ and a sensor design $\pi_1$, constraint~\eqref{constraint: worst-case regret attacker} restricts the optimal attacker's value $V_{2, i}$ satisfies the Bellman optimality condition, and $\pi_2^{i}$ is the corresponding optimal policy. Constraint~\eqref{constraint: worst-case regret defender} restricts the optimal defender's value $V_{1, i}$ to satisfy the Bellman optimality condition.
Constraints~\eqref{constraint: defender regret}, \eqref{constraint: worst-case regret attacker} and\eqref{constraint: worst-case regret defender} together enforces $y \ge \max_{i \in \types} (\bar v_{1, i} - \sum_{s\in S} \nu(s) V_{1, i}(s))$.

The rest of the reasoning follows the same from the zero-sum case and is repeated here for completeness: 
For an arbitrary $\pi_1$, let $r(\pi_1) = \argmax_{i\in \types} ( \bar v_{1, i} - \sum_{s\in S} \nu(s) V_{1, i}(s))$, that is the $r(\pi_1)$ is the attacker type for which the defender's regret of using sensor design $\pi_1$ is the largest among the regret for all attacker's types. Then we have $y \ge \max_{i\in \types} (\bar v_{1, i} - \sum_{s\in S} \nu(s) V_{1, i}(s))  =  \bar v_{1, r(\pi_1)} - \sum_{s\in S} \nu(s) V_{1, r(\pi_1)}(s)$. Because $\bar v_{1, r(\pi_1)}$ is a constant once $r(\pi_1)$ is determined. Minimizing $y$ is equivalent to maximizing the upperbound of $\sum_{s\in S} \nu(s) V_{1, r(\pi_1)}(s)$ and thus $y = \bar v_{1, r(\pi_1)} - \sum_{s\in S} \nu(s) V_{1, r(\pi_1)}(s)$. The optimization problem is then $\min_{\pi_1} (v_{1, r(\pi_1)} - \sum_{s\in S} \nu(s) V_{1, r(\pi_1)}(s))$.

\subsection{False positive/negative rate of sensors}
If the sensor has a non-zero false positive or false negative rate, then we can modify the transition probability $P$ to solve the problem. We first clarify that the false positive rate does not change the transition probability. This is because if the attack \ac{mdp} only captures how the attacker makes progress by taking attack actions, such as exploiting known network vulnerabilities. The normal user's actions are not considered and thus the transition function does not consider a false positive rate. In particular applications, false negatives pose greater risks to system security than false positives.

Therefore, we only need to consider the false negative rate, denoted by $\epsilon$. In the zero-sum game, the term  $\sum_{s'}P^{\vec{x}}(s'|s,a)V_{2,i}(s';\vec{x})$ in~\eqref{eq:robust_constraint_lp} satisfies
\resizebox{.98\linewidth}{!}{
\begin{minipage}{\linewidth}

\begin{align*}
 & \sum_{s'}P^{\vec{x}}(s'|s,a)V_{2,i}(s';\vec{x})\\
& =\begin{cases}
(1- \epsilon)V_{2,i}(s_{\sink};\vec{x}) + \epsilon \sum\limits_{s'}P(s'|s,a)V_{2,i}(s';\vec{x}), & \vec{x}(s)=1,\\
\sum\limits_{s'}P(s'|s,a)V_{2,i}(s';\vec{x}), & \vec{x}(s)=0
\end{cases}\\
&=\sum_{s'}P(s'|s,a)V_{2,i}(s';\vec{x})(1-\vec{x}(s)+\epsilon \vec{x}(s)),
\end{align*}
\end{minipage}
}
then follow the big-M method we used in the zero-sum game. Addressing the false-negative rate in the non-zero-sum game is straightforward: we simply modify the McCormick relaxation part by changing $\mathcal{P}(s'|s, 1, a_2)$ to
\begin{equation*}
\label{eq:falsenegative_px}
\mathcal{P}(s'|s, 1, a_2)=
\begin{cases}
	1-\epsilon, &\ s'=s_{\sink},\\
	\epsilon P(s'|s,a_2), & s'\neq s_{\sink}.
\end{cases} 
\end{equation*}


\subsection{Scalability analysis regarding the number of attackers}
The complexity analysis shows that the number of integer variables grows linearly in the number of attacker types. To evaluate how the method scales with respect to the number of attackers.  Let's consider the zero-sum case and the two attacker types mentioned in the main draft. Assuming that attacker 3's rewards upon reaching Cash, Gold, and Diamond are 15, 12, and 15 respectively, and attacker 4's rewards upon reaching these locations are 12, 12, and 15 respectively. When only attacker types 1 and 2 are present, the running time is 0.016 seconds. When potential attacker types are restricted to 1, 2, and 3, the running time increases to 0.997 seconds. If we include attacker type 4, the running time further increases to 0.875 seconds. Finally, when all four potential attacker types are considered, the running time is 1.130 seconds. From these experiments, we observe that the running time does not grow as a linear function of the attacker numbers. Even for the same number of attackers, the attacker's reward function also influences the running time. 
\end{document}